\documentclass{llncs}
\usepackage{amsmath,amssymb}
\usepackage{tabls}
\usepackage{graphicx}
\usepackage{wrapfig}
\usepackage{array}
\usepackage[algoruled,linesnumbered,algo2e,vlined]{algorithm2e}
\usepackage{url}
\usepackage{amssymb}
\usepackage{multirow}
\usepackage{multicol}
\usepackage[final,mode=multiuser]{fixme}
\fxuseenvlayout{colorsig}
\fxusetargetlayout{color}
\FXRegisterAuthor{ti}{ati}{TI}
\FXRegisterAuthor{hb}{ahb}{HB}
\FXRegisterAuthor{si}{asi}{SI}
\FXRegisterAuthor{mt}{amt}{MT}
\definecolor{kgnote}{rgb}{1.0000,0.0000,0.0000}
\fxsetface{env}{}

\newcommand{\Occ}{\mathit{Occ}}

\newcommand{\LCP}{\mathit{lcp}}
\newcommand{\LCS}{\mathit{lcs}}

\newcommand{\plen}{\mathit{plen}}
\newcommand{\slen}{\mathit{slen}}

\newcommand{\triple}[3]{\langle#1,#2,#3\rangle}
\newcommand{\quintuplet}[5]{\langle#1,#2,#3,#4,#5\rangle}
\newcommand{\gpal}[2]{\langle#1,#2\rangle_{g}}
\newcommand{\Run}{\mathit{Run}}

\newcommand{\kushiRun}{\mathit{Run}^{\xi}}
\newcommand{\gpals}{g\mathit{Pals}}
\newcommand{\kushigpals}{g\mathit{Pals}^{\xi}}
\newcommand{\overl}[1]{\overleftarrow{#1}}
\newcommand{\overr}[1]{\overrightarrow{#1}}
\newcommand{\repl}[2]{\overl{\mathit{rep}}_{#2}({#1})}
\newcommand{\repr}[2]{\overr{\mathit{rep}}_{#2}({#1})}
\newcommand{\rev}[1]{#1^R}

\newlength\savedwidth

\newcommand{\height}{\mathit{height}}

\newcommand{\LS}{\mathsf{LS}}
\newcommand{\LCE}{\mathsf{LCE}}
\newcommand{\Match}{\mathsf{Match}}
\newcommand{\FirstMismatch}{\mathsf{FirstMismatch}}
\newcommand{\OccKushi}{\mathit{Occ^{\xi}}}

\newcommand{\Oh}[1]
    {\ensuremath{\mathcal{O}\!\left( {#1} \right)}}
\newcommand{\derive}{\ensuremath{\mathit{val}}}

\author{
}
\institute{
  Department of Informatics, Kyushu University\\
  \email{@inf.kyushu-u.ac.jp}\\
}
\title{
  Detecting regularities on grammar-compressed strings
}

 \author{
   Tomohiro I\inst{1,2}
   \and
   Wataru Matsubara\inst{3}
   \and
   Kouji Shimohira\inst{1}
   \and\\
   Shunsuke Inenaga\inst{1}
   \and 
   Hideo Bannai\inst{1}
   \and 
   Masayuki Takeda\inst{1}
   \and \\
   Kazuyuki Narisawa\inst{3} 
   \and 
   Ayumi Shinohara\inst{3}
 }

 \institute{
   Department of Informatics, Kyushu University, Japan \\
   \email{\{tomohiro.i, inenaga, bannai, takeda\}@inf.kyushu-u.ac.jp}
   \and
   Japan Society for the Promotion of Science (JSPS)
   \and
   Graduate School of Information Sciences, Tohoku University, Japan \\
   \email{\{narisawa,ayumi\}@ecei.tohoku.ac.jp}
 }

\begin{document}
\maketitle

\begin{abstract}
We solve the problems of detecting and counting various forms of
regularities in a string represented as a Straight Line Program (SLP).
Given an SLP of size $n$ that represents a string $s$ of length $N$,
our algorithm compute all runs and squares in $s$
in $O(n^3h)$ time and $O(n^2)$ space,
where $h$ is the height of the derivation tree of the SLP.
We also show an algorithm to compute all gapped-palindromes
in $O(n^3h + gnh\log N)$ time and $O(n^2)$ space, 
where $g$ is the length of the gap.
The key technique of the above solution also allows us to compute
the periods and covers of the string in $O(n^2 h)$ time
and $O(nh(n+\log^2 N))$ time, respectively.
\end{abstract}

\section{Introduction} \label{sec:introduction}

Finding regularities such as squares, runs, and palindromes
in strings, is a fundamental and important problem in
stringology with various applications,
and many efficient algorithms have been proposed
(e.g.,~\cite{MainLorentz84,crochemore91:_effic,AB96,jansson07:_onlin_dynam_recog_squar_strin,Manacher75,ApostolicoBG95,KolpakovK99}).
See also~\cite{CrochemoreIR09} for a survey.

In this paper, we consider the problem of detecting
regularities in a string $s$ of length $N$ that is given in a compressed form,
namely, as a straight line program (SLP),
which is essentially a context free grammar in the Chomsky normal form
that derives only $s$.
Our model of computation is the word RAM:
We shall assume that the computer word size is at least $\lceil \log_2 N \rceil$, 
and hence, standard operations on
values representing lengths and positions of string $s$
can be manipulated in constant time.
Space complexities will be determined by the number of computer words (not bits).

Given an SLP whose size is $n$ and the height of its derivation tree is $h$,
Bannai et al.~\cite{Bannai2012eat} showed how to test 
whether the string $s$ is square-free or not, in $O(n^3 h \log N)$ time and $O(n^2)$ space.
Independently, Khvorost~\cite{Khvorost2012CAS} presented an algorithm for computing a compact representation of all squares in $s$
in $O(n^3 h \log^2 N)$ time and $O(n^2)$ space.
Matsubara et al.~\cite{matsubara_tcs2009} showed that 
a compact representation of all maximal palindromes occurring in the string $s$
can be computed in $O(n^3 h)$ time and $O(n^2)$ space.
Note that the length $N$ of the decompressed string 
$s$ can be as large as $O(2^n)$ in the worst case.
Therefore, in such cases
these algorithms are more efficient than \emph{any} algorithm
that work on uncompressed strings.

In this paper we present the following extension and improvements to the above work, 
namely,
\begin{enumerate}
\item \label{result:runs} an $O(n^3h)$-time $O(n^2)$-space algorithm for computing
a compact representation of squares and runs;
\item \label{result:gapped_pals} 
an $O(n^3h + gnh\log N)$-time $O(n^2)$-space algorithm for computing
a compact representation of palindromes with a gap (spacer) of length $g$.
\end{enumerate}
We remark that our algorithms can easily be extended to count the number of 
squares, runs, and gapped palindromes in the same 
time and space complexities.

\sinote*{added}{%
Note that Result~\ref{result:runs} improves on the work 
by Khvorost~\cite{Khvorost2012CAS} which requires $O(n^3 h \log^2 N)$ time and $O(n^2)$ space.
The key to the improvement is our new technique of Section~\ref{subsec:approximate_doubling}
called \emph{approximate doubling},
which we believe is of independent interest.
In fact, using the approximate doubling technique,
one can improve the time complexity of the algorithms of Lifshits~\cite{lifshits07:_proces_compr_texts}
to compute the periods and covers of a string given as an SLP,
in $O(n^2h)$ time and $O(nh(n+\log^2N))$ time, respectively.
}%

\sinote*{added}{%
If we allow no gaps in palindromes (i.e., if we set $g = 0$),
then Result~\ref{result:gapped_pals} implies that we can compute 
a compact representation of all maximal palindromes in $O(n^3h)$ time and $O(n^2)$ space.
Hence, Result~\ref{result:gapped_pals} can be seen as a generalization of 
the work by Matsubara et al.~\cite{matsubara_tcs2009} with the same efficiency.
}%

\section{Preliminaries}

\subsection{Strings}

Let $\Sigma$ be the alphabet,
so an element of $\Sigma^*$ is called a string.
For string $s = xyz$,
$x$ is called a prefix,
$y$ is called a substring,
and $z$ is called a suffix of $s$, respectively.
The length of string $s$ is denoted by $|s|$.
The empty string $\varepsilon$ is a string of length 0,
that is, $|\varepsilon| = 0$.
For $1 \leq i \leq |s|$, $s[i]$ denotes the $i$-th character of $s$.
For $1 \leq i \leq j \leq |s|$,
$s[i..j]$ denotes the substring of $s$ that begins at position $i$
and ends at position $j$.
For any string $s$, let $\rev{s}$ denote the reversed string of $s$,
that is, $\rev{s} = s[|s|] \cdots s[2]s[1]$.
For any strings $s$ and $u$, 
let $\LCP(s, u)$ (resp. $\LCS(s, u)$) denote the length of the longest common prefix (resp. suffix)
of $s$ and $u$.

We say that string $s$ has a \emph{period} $c$ ($0 < c \leq |s|$) 
if $s[i] = s[i+c]$ for any $1 \leq i \leq |s|-c$.
For a period $c$ of $s$, we denote $s = u^q$, 
where $u$ is the prefix of $s$ of length $c$ and $q = \frac{|s|}{c}$.
For convenience, let $u^0 = \varepsilon$.
If $q \geq 2$, $s = u^q$ is called a \emph{repetition} with root $u$ and period $|u|$.
Also, we say that $s$ is \emph{primitive} if there is no string $u$ and integer $k > 1$ 
such that $s = u^k$.
If $s$ is primitive, then $s^2$ is called a \emph{square}.

We denote a repetition in a string $s$ by a triple $\triple{b}{e}{c}$ 
such that $s[b..e]$ is a repetition with period $c$.
A repetition $\triple{b}{e}{c}$ in $s$ is called a \emph{run} (or \emph{maximal periodicity} in \cite{Main1989Dlm}) if
$c$ is the smallest period of $s[b..e]$ and the substring cannot be extended to the left nor to the right with the same period, namely
neither $s[b-1..e]$ nor $s[b..e+1]$ has period $c$.
Note that for any run $\triple{b}{e}{c}$ in $s$, 
every substring of length $2c$ in $s[b..e]$ is a square.
Let $\Run(s)$ denote the set of all runs in $s$.

A string $s$ is said to be a \emph{palindrome}
if $s = \rev{s}$.
\sinote*{redefined}{%
A string $s$ said to be a \emph{gapped} palindrome
if $s = x u \rev{x}$ for some string $u \in \Sigma^*$.
Note that $u$ may or may not be a palindrome.
The prefix $x$ (resp. suffix $\rev{x}$) of $xu\rev{x}$
is called the \emph{left arm} (resp. \emph{right arm}) of gapped palindrome $x u \rev{u}$.
If $|u| = g$, then $x u \rev{x}$ is said to be a $g$-gapped palindrome.
}%
We denote a \emph{maximal $g$-gapped palindrome} in a string $s$ by a pair $\gpal{b}{e}$
such that $s[b..e]$ is a $g$-gapped palindrome and $s[b-1..e+1]$ is not.
Let $\gpals(s)$ denote the set of all maximal $g$-gapped palindromes in $s$.

Given a text string $s \in \Sigma^+$ and a pattern string $p \in \Sigma^+$, 
we say that $p$ occurs at position $i$ ($1 \leq i \leq |s| - |p| + 1$) iff $s[i..i+|p|-1] = p$.
Let $\Occ(s, p)$ denote the set of positions where $p$ occurs in $s$.
For a pair of integers $1 \leq b \leq e$, 
$[b, e] = \{b, b+1, \dots, e\}$ is called an \emph{interval}.
\begin{lemma}[\cite{MST97}]\label{lem:ap}
For any strings $s, p \in \Sigma^+$ and any interval $[b, e]$ with $1 \leq b \leq e \leq b + |p|$, 
$\Occ(s, p) \cap [b, e]$ forms a single arithmetic progression if $\Occ(s, p) \cap [b, e] \neq \emptyset$.
\end{lemma}

\subsection{Straight-line programs}

A \emph{straight-line program} (\emph{SLP}) $\mathcal{S}$ of size $n$
is a set of productions
$\mathcal{S} = \{X_i \rightarrow \mathit{expr}_i\}_{i = 1}^{n}$,
where each $X_i$ is a distinct variable and
each $\mathit{expr}_i$ is either
$\mathit{expr_i} = X_\ell X_r~(1 \leq \ell, r < i)$,
or $\mathit{expr_i} = a$ for some $a \in \Sigma$.
Note that $X_n$ derives only a single string and, therefore, we
view the SLP as a compressed representation of the string $s$ 
that is derived from the variable $X_n$.
Recall that the length $N$ of the string $s$ can be as large as $O(2^n)$.
However, it is always the case that $n \geq \log N$.
For any variable $X_i$, let $\derive(X_i)$ denote the string that is
derived from variable $X_i$.
Therefore, $\derive(X_n) = s$.
When it is not confusing, 
we identify $X_i$ with the string represented by $X_i$.

Let $T_i$ denote the derivation tree of a variable $X_i$ of an SLP $\mathcal{S}$.
The derivation tree of $\mathcal{S}$ is $T_n$ (see also Fig.~\ref{fig:SLP} in Appendix C).
Let $\height(X_i)$ denote the height of the derivation tree $T_i$ of $X_i$ 
and $\height(\mathcal{S}) = \height(X_n)$.
We associate each leaf of $T_i$
with the corresponding position of the string $\derive(X_i)$.
For any node $z$ of the derivation tree $T_{i}$,
let $\ell_z$ be the number of leaves to the left of $z$ in $T_{i}$.
The position of $z$ in $T_i$ is $\ell_z + 1$.

Let $[u,v]$ be any integer interval with $1 \leq u \leq v \leq |\derive(X_i)|$.
We say that the interval $[u,v]$ \emph{crosses the boundary} of node $z$ in $T_{i}$,
if the lowest common ancestor of the leaves $u$ and $v$ in $T_{i}$ is $z$.
We also say that the interval $[u, v]$ \emph{touches the boundary} of node $z$ in $T_i$,
if either $[u-1, v]$ or $[u, v+1]$ crosses the boundary of $z$ in $T_i$.
Assume $p = w[u..u+|p|-1]$ and interval $[u,u+|p|-1]$ crosses or touches the boundary of
node $z$ in $T_i$.
When $z$ is labeled by $X_j$, then we also say that 
the occurrence of $p$ starting
at position $u$ in $\derive(X_i)$ crosses or touches the boundary of $X_j$.

\begin{lemma}[\cite{philip11:_random_acces_gramm_compr_strin}] \label{lem:random_access}
Given an SLP $\mathcal{S}$ of size $n$ describing string $w$ of length $N$,
we can pre-process $\mathcal{S}$ in $O(n)$ time and space
to answer the following queries in $O(\log N)$ time:
\begin{itemize}
\item Given a position $u$ with $1 \leq u \leq N$, answer the character $w[u]$.
\item Given an interval $[u,v]$ with $1 \leq u \leq v \leq N$, 
      answer the node $z$ the interval $[u,v]$ crosses,
      the label $X_i$ of $z$, 
      and the position of $z$ in $T_{\mathcal{S}} = T_n$.
\end{itemize}
\end{lemma}

%
%
For any production $X_i \rightarrow X_\ell X_r$
and a string $p$,
let $\OccKushi(X_i, p)$ be the set of occurrences of $p$
which begin in $X_\ell$ and end in $X_r$.
Let $\mathcal{S}$ and $\mathcal{T}$ be SLPs of sizes $n$ and $m$, respectively.
Let the AP-table for $\mathcal{S}$ and $\mathcal{T}$ be an $n \times m$ table
such that for any pair of variables $X \in \mathcal{S}$ and $Y \in \mathcal{T}$
the table stores $\OccKushi(X, Y)$.
It follows from Lemma~\ref{lem:ap} that $\OccKushi(X, Y)$ 
forms a single arithmetic progression which requires $O(1)$ space,
and hence the AP-table can be represented in $O(nm)$ space.

\begin{lemma}[\cite{lifshits07:_proces_compr_texts}] \label{lem:SLP_matching}
Given two SLPs $\mathcal{S}$ and $\mathcal{T}$ of sizes $n$ and $m$,
respectively, the AP-table for $\mathcal{S}$ and $\mathcal{T}$ can be
computed in $O(nmh)$ time and $O(nm)$ space,
where $h = \height(\mathcal{S})$.
\end{lemma}

\begin{lemma}[\cite{lifshits07:_proces_compr_texts}, local search ($\LS$)]\label{lem:local_search}
Using AP-table for $\mathcal{S}$ and $\mathcal{T}$ that describe strings $p$ in $s$,
we can compute, given any position $b$ and constant $\alpha > 0$, 
$\Occ(s, p) \cap [b, b + \alpha |p|]$ as a form of at most $\lceil \alpha \rceil$ arithmetic progressions
in $O(h)$ time, where $h = \height(\mathcal{S})$.
\end{lemma}

Note that, given any $1 \leq i \leq j \leq |s|$, 
we are able to build an SLP of size $O(n)$ that generates substring $s[i..j]$ in $O(n)$ time.
Hence, by computing the AP-table for $\mathcal{S}$ and the new SLP,
we can conduct the local search $\LS$ operation on substring $s[i..j]$ in $O(n^2h)$ time.

For any variable $X_i$ of $\mathcal{S}$ and positions $1 \leq k_1, k_2 \leq |X_i|$,
we define the ``right-right'' longest common extension query by
\[\LCE(X_i, k_1, k_2) = \LCP(X_i[k_1..|X_i|], X_i[k_2..|X_i|]).\]
Using a technique of~\cite{MST97} in conjunction with Lemma~\ref{lem:SLP_matching},
it is possible to answer the query in $O(n^2h)$ time
for each pair of positions, with no pre-processing.
We will later show our new algorithm which, after $O(n^2h)$-time pre-processing,
answers to the $\LCE$ query for any pair of positions in $O(h \log N)$ time.

\section{Finding runs} \label{sec:repetitions}

In this section we propose an $O(n^3h)$-time and $O(n^2)$-space algorithm to compute $O(n \log N)$-size representation of all runs
in a text $s$ of length $N$ represented by SLP $\mathcal{S} = \{X_i \rightarrow \mathit{expr_i}\}_{i = 1}^{n}$ of height $h$.

For each production $X_i \rightarrow X_{\ell(i)} X_{r(i)}$ with $i \leq n$,
we consider the set $\kushiRun(X_i)$ of runs which touch or cross the boundary of $X_i$ and are completed in $X_i$,
i.e., those that are not prefixes nor suffixes of $X_i$.
Formally, 
\[
\kushiRun(X_i) = \{\triple{b}{e}{c} \in \Run(X_i) \mid 1 \leq b-1 \leq |X_{\ell(i)}| < e+1 \leq |X_i| \}.
\]
It is known that for any interval $[b, e]$ with $1 \leq b \leq e \leq |s|$,
there exists a unique occurrence of a variable $X_i$ 
in the derivation tree of SLP, such that 
the interval $[b, e]$ crosses the boundary of $X_i$.
Also, wherever $X_i$ appears in the derivation tree, 
the runs in $\kushiRun(X_i)$ occur in $s$ with some appropriate offset,
and these occurrences of the runs are never contained in $\kushiRun(X_j)$ with 
any other variable $X_j$ with $j \neq i$.
Hence, by computing $\kushiRun(X_i)$ for all variables $X_i$ with $i \leq n$,
we can essentially compute all runs of $s$ that are not prefixes nor suffixes of $s$.
In order to detect prefix/suffix runs of $s$, 
it is sufficient to consider two auxiliary variables 
$X_{n+1} \rightarrow X_{\$} X_n$ and $X_{n+2} \rightarrow X_{n+1} X_{\$^{\prime}}$, 
where $X_{\$}$ and $X_{\$^{\prime}}$ respectively derive special characters $\$$ and $\$^{\prime}$ that are not in $s$ and $\$ \neq \$^{\prime}$.
Hence, the problem of computing the runs from an SLP $\mathcal{S}$ reduces to computing
$\kushiRun(X_i)$ for all variables $X_i$ with $i \leq n + 2$.

Our algorithm is based on the divide-and-conquer method used in~\cite{Bannai2012eat} and also~\cite{Khvorost2012CAS},
which detect squares crossing the boundary of each variable $X_i$.
Roughly speaking, in order to detect such squares
we take some substrings of $\derive(X_i)$ as \emph{seeds}
each of which is in charge of distinct squares,
and for each seed we detect squares by using $\LS$ and $\LCE$ constant times.
There is a difference between~\cite{Bannai2012eat} and~\cite{Khvorost2012CAS}
in how the seeds are taken, and ours is rather based on that in~\cite{Bannai2012eat}.
In the next subsection, we briefly describe our basic algorithm which runs in $O(n^3h \log N)$ time.

\subsection{Basic algorithm} \label{subsec:basic_algorithm}

Consider runs in $\kushiRun(X_i)$ with $X_i \rightarrow X_\ell X_r$.
Since a run in $\kushiRun(X_i)$ contains a square which touches or crosses the boundary of $X_i$,
our algorithm finds a run by first finding such a square,
and then computing the maximal extension of its period to the left and right of its occurrence.

We divide each square $ww$ by its length and how it relates to the boundary of $X_i$.
When $|w| > 1$, there exists $1 \leq t < \log|\derive(X_{i})|$ such that $2^{t} \leq |w| < 2^{t+1}$
and there are four cases (see also Fig.~\ref{fig:cases});
(1) $|w_\ell| \geq \frac{3}{2} |w|$,
(2) $\frac{3}{2} |w| > |w_\ell| \geq |w|$,
(3) $|w| > |w_\ell| \geq \frac{1}{2} |w|$,
(4) $\frac{1}{2} |w| > |w_\ell|$,
where $w_\ell$ is a prefix of $ww$ which is also a suffix of $\derive(X_\ell)$.

\begin{figure*}[t]
  \begin{center}
    \includegraphics[scale=0.70]{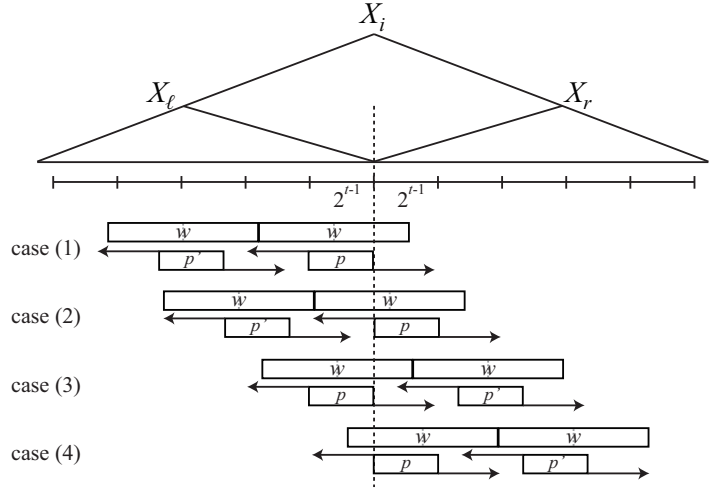}
  \end{center}
  \caption{
    The left arrows represent the longest common suffix between the left substrings immediately to the left of $p$ and $p'$.
    The right arrows represent the longest common prefix between the substrings immediately to the right of $p$ and $p'$.
  }
  \label{fig:cases}
\end{figure*}

The point is that in any case we can take a substring $p$ of length $2^{t-1}$ 
of $s$ which touches the boundary of $X_i$, and is completely contained in $w$.
By using $p$ as a seed we can detect runs by the following steps:
\begin{description}
  \item[Step 1:] Conduct local search of $p$ in an ``appropriate range'' of $X_{i}$,
        and find a copy $p'~(= p)$ of $p$. \label{enum:matching}
  \item[Step 2:] Compute the length $\plen$ of the longest common prefix to the right of $p$ and $p'$,
        and the length $\slen$ of the longest common suffix to the left of $p$ and $p'$, then
        check that $\plen + \slen \geq d - |p|$, 
        where $d$ is the distance between the beginning positions of $p$ and $p'$. \label{enum:lcp}
\end{description}
Notice that Step 2 actually computes maximal extension of the repetition.

Since $d = |w|$, 
it is sufficient to conduct local search in the range satisfying $2^{t} \leq d < 2^{t+1}$, namely,
the width of the interval for local search is smaller than $2|p|$, and all occurrences of $p'$ are represented by at most two arithmetic progressions.
Although exponentially many runs can be represented by an arithmetic progression,
its periodicity enables us to efficiently detect all of them, by using $\LCE$ only constant times, 
and they are encoded in $O(1)$ space.
We put the details in Appendix A since the employed techniques are essentially 
the same as in~\cite{Khvorost2012CAS}.

By varying $t$ from $1$ to $\log N$, we can obtain an $O(\log N)$-size compact representation of $\kushiRun(X_i)$ in $O(n^2h \log N)$ time.
More precisely, we get a list of 
$O(\log N)$ quintuplets $\quintuplet{\delta_1}{\delta_2}{\delta_3}{c}{k}$ such that
the union of sets $\bigcup_{j = 0}^{k-1}\triple{\delta_1-cj}{\delta_2+cj}{\delta_3+cj}$ 
for all elements of the list equals to $\kushiRun(X_i)$ without duplicates.
By applying the above procedure to all the $n$ variables, 
we can obtain an $O(n \log N)$-size compact representation of all runs in $s$ in $O(n^3h \log N)$ time.
The total space requirement is $O(n^2)$,
since we need $O(n^2)$ space at each step of the algorithm.

In order to improve the running time of the algorithm to $O(n^3h)$,
we will use new techniques of the two following subsections.

\subsection{Longest common extension} \label{sec:lce}

In this subsection we propose a more efficient algorithm for $\LCE$ queries.
\begin{lemma} \label{lem:LCE}
We can pre-process an SLP $\mathcal{S}$ of size $n$
and height $h$ in $O(n^2 h)$ time and $O(n^2)$ space, 
so that given any variable $X_i$ and positions $1 \leq k_1, k_2 \leq |X_i|$,
$\LCE(X_i, k_1, k_2)$ is answered in $O(h \log N)$ time.
\end{lemma}

To compute $\LCE(X_i, k_1, k_2)$
we will use the following function:
For an SLP $\mathcal{S} = \{X_i \rightarrow expr_i\}_{i=1}^{n}$,
let $\Match$ be a function such that 
\[
  \Match(X_i, X_j, k) = 
   \begin{cases}
    \mathrm{true} & \mathrm{if } \ k \in \Occ(X_i, X_j), \\
    \mathrm{false} & \mathrm{if } \ k \notin \Occ(X_i, X_j). \\
   \end{cases}
\]
\begin{lemma} \label{lem:match}
We can pre-process a given SLP $\mathcal{S}$ of size $n$ and height $h$ 
in $O(n^2 h)$ time and $O(n^2)$ space 
so that the query $\Match(X_i, X_j, k)$ is  
answered in $O(\log N)$ time.
\end{lemma}

\begin{proof}
We apply Lemma~\ref{lem:random_access} to \emph{every} variable $X_i$ of $\mathcal{S}$,
so that the queries of Lemma~\ref{lem:random_access} is answered in 
$O(\log N)$ time on the derivation tree $T_i$ of each variable $X_i$ of $\mathcal{S}$.
Since there are $n$ variables in $\mathcal{S}$, this takes a total of 
$O(n^2)$ time and space.
We also apply Lemma~\ref{lem:SLP_matching} to $\mathcal{S}$,
which takes $O(n^2 h)$ time and $O(n^2)$ space.
Hence the pre-processing takes a total of $O(n^2 h)$ time and $O(n^2)$ space.

To answer the query $\Match(X_i, X_j, k)$, 
we first find the node of $T_{i}$ the interval $[k, k+|X_j|-1]$ crosses,
its label $X_q$, and its position $r$ in $T_{i}$.
This takes $O(\log N)$ time using Lemma~\ref{lem:random_access}.
Then we check in $O(1)$ time if $(k-r) \in \OccKushi(X_q, X_j)$ or not, 
using the arithmetic progression stored in the AP-table.
Thus the query is answered in $O(\log N)$ time.
\qed
\end{proof}

The following function will also be used in our algorithm:
Let $\FirstMismatch$ be a function such that 
\[
  \FirstMismatch(X_i, X_j, k) = 
   \begin{cases}
    |\LCP(X_i[k..|X_i|], X_j)| & \mathrm{if} \ |X_i|-k+1 \leq |X_j|, \\
    \mathrm{undefined} & \mathrm{otherwise.}
   \end{cases}
\]

Using Lemma~\ref{lem:match} we can establish the following lemma.
See Appendix B for a full proof.
\begin{lemma} \label{lem:fm}
We can pre-process a given SLP $\mathcal{S}$ of size $n$ and height $h$
in $O(n^2 h)$ time and $O(n^2)$ space 
so that the query $\FirstMismatch(X_i, X_j, k)$ is
answered in $O(h \log N)$ time.
\end{lemma}

%
%
%

%
%

We are ready to prove Lemma~\ref{lem:LCE}:
\begin{proof}
Consider to compute $\LCE(X_i, k_1, k_2)$.
Without loss of generality, assume $k_1 \leq k_2$.
Let $z$ be the lca of the $k_1$-th and $(k_2 - k_1 + |X_i|)$-th leaves
of the derivation tree $T_i$.
Let $P_\ell$ be the path from $z$ to the $k_1$-th leaf
of the derivation tree $T_i$, and 
let $L$ be the list of the right child of the nodes in $P_\ell$
sorted in increasing order of their position in $T_i$.
The number of nodes in $L$ is at most $\height(X_i) \leq h$,
and $L$ can be computed in $O(\height(X_i)) = O(h)$ time.
Let $P_r$ be the path from $z$ to the $(k_2 - k_1 + |X_i|)$-th leaf 
of the derivation tree $T_i$,
and let $R$ be the list of the left child of the nodes in $P_r$
sorted in increasing order of their position in $T_i$.
$R$ can be computed in $O(h)$ time as well.
Let $U = L \cup R = \{X_{u(1)}, X_{u(2)}, \ldots, X_{u(m)}\}$ 
be the list obtained by concatenating $L$ and $R$.
For each $X_{u(p)}$ in increasing order of $p = 1, 2, \ldots, m$,
we perform query $\Match(X_i, X_{u(p)}, k_1 + \sum_{q = 1}^{p-1}|X_{u(q)}|)$ 
until either finding the first variable $X_{u(p^\prime)}$ for which the query returns false
(see also Fig.~\ref{fig:lce} in Appendix C),
or all the queries for $p = 1, \ldots, m$ have returned true.
In the latter case, clearly $\LCE(X_i, k_1, k_2) = |X_i| - k_1 + 1$.
In the former case, the first mismatch occurs between $X_i$ and $X_{u(p^\prime)}$, 
and hence 
$\LCE(X_i, k_1, k_2) = \sum_{q^\prime = 1}^{p^\prime-1}|X_{u(q^\prime)}| + \FirstMismatch(X_i, X_{u(p^\prime)}, k_1 + \sum_{q^\prime = 1}^{p^\prime-1} |X_{u(q^\prime)}|)$.

Since $U$ contains at most $2 \cdot \height(X_i)$ variables,
we perform  $O(h)$ $\Match$ queries.
We perform at most one $\FirstMismatch$ query.
Thus, using Lemmas~\ref{lem:match} and~\ref{lem:fm},
we can compute $\LCE(X_i, k_1, k_2)$ in $O(h \log N)$ time
after $O(n^2 h)$-time $O(n^2)$-space pre-processing.
\qed
\end{proof}

%
%

We can use Lemma~\ref{lem:LCE} to also compute 
``left-left'',
``left-right'',
and ``right-left'' longest common extensions on the uncompressed string $s = \derive(\mathcal{S})$:
We can compute in $O(n)$ time an SLP $\mathcal{S}^R$
of size $n$ which represents the reversed string $s^R$~\cite{matsubara_tcs2009}.
We then construct a new SLP $\mathcal{S}^\prime$ of size $2n$ and height $h+1$
by concatenating the last variables of $\mathcal{S}$ and $\mathcal{S}^R$,
and apply Lemma~\ref{lem:LCE} to $\mathcal{S}^\prime$.

\subsection{Approximate doubling} \label{subsec:approximate_doubling}

Here we show how to reduce the number of AP-table computation required 
in Step 1 of the basic algorithm, from $O(\log N)$ to $O(1)$ times per variable.

Consider any production $X_i \rightarrow X_\ell X_r$. 
If we build a new SLP which contains variables
that derive the prefixes of length $2^t$ of $X_r$ for each $0 \leq t < \log |X_r|$,
we can obtain the AP-tables for $X_i$ and all prefix seeds of $X_r$ by computing the AP-table for $X_i$ and the new SLP.
Unfortunately, however, the size of such a new SLP can be as large as $O(n \log N)$.
Here we notice that the lengths of the seeds do not have to be exactly doublings,
i.e., the basic algorithm of Section~\ref{subsec:basic_algorithm} works fine
as long as the following properties are fulfilled:
(a) the ratio of the lengths for each pair of consecutive seeds is constant;
(b) the whole string is covered by the $O(\log N)$ seeds
\footnote{A minor modification is that we conduct local search for a seed $p$ at Step 1
with the range satisfying $2|p| \leq d < 2|q|$, where $q$ is the next longer seed of $p$.}.
We show in the next lemma that we can build an approximate doubling SLP of size $O(n)$.

\begin{lemma}\label{lem:approx_doubling}
Let $\mathcal{S} = \{X_i \rightarrow \mathit{expr}_i\}_{i=1}^{n}$ 
be an SLP that derives a string $s$.
We can build in $O(n)$ time a new SLP $\mathcal{S'} = \{Y_i \rightarrow \mathit{expr'}_i\}_{i=1}^{n'}$
with $n' = O(n)$ and $\height(\mathcal{S'}) = O(\height(\mathcal{S}))$,
which derives $s$ and
contains $O(\log N)$ variables $Y_{a_1}, Y_{a_2}, \dots, Y_{a_k}$ satisfying the following conditions:
\begin{itemize}
\item For any $1 \leq j \leq k$, $Y_{a_j}$ derives a prefix of $s$, $|Y_{a_1}| = 1$ and $|Y_{a_k}| = |s|$.
\item For any $1 \leq j < k$, $|Y_{a_{j}}| < |Y_{a_{j+1}}| \leq 2|Y_{a_{j}}|$.
\end{itemize}
\end{lemma}

\begin{proof}
First, we copy the productions of $\mathcal{S}$ into $\mathcal{S'}$.
Next we add productions needed for creating prefix variables $Y_{a_1}, Y_{a_2}, \dots, Y_{a_k}$ in increasing order.
We consider separating the derivation tree $T_n$ of $X_n$ into segments by a sequence of nodes $v_1, v_2, \dots, v_{k}$ such that
the $i$-th segment enclosed by the path from $v_i$ to $v_{i+1}$ represents the suffix of $Y_{a_{i+1}}$ of length $|Y_{a_{i+1}}| - |Y_{a_{i}}|$,
namely, $Y_{a_{i+1}} \rightarrow Y_{a_{i}} Y_{b_{i}}$ where $Y_{b_{i}}$ is a variable for the $i$-th segment.
Each node $v_i$ is called an l-node (resp. r-node) if the node belongs to the left (resp. right) segment of the node.

We start from $v_1$ which is the leftmost node that derives $s[1]$.
Suppose we have built prefix variables up to $Y_{a_{i}}$ and now creating $Y_{a_{i+1}}$.
At this moment we are at $v_{i}$.
We move up to the node $u_{i}$ such that $u_{i}$ is the deepest node on the path from the root to $v_{i}$ which contains position $2|Y_{a_{i}}|$,
and move down from $u_{i}$ towards position $2|Y_{a_{i}}|$.
The traversal ends when we meet a node $v_{i+1}$ which satisfies one of the following conditions;
(1) the rightmost position of $v_{i+1}$ is $2|Y_{a_{i}}|$,
(2) $v_{i+1}$ is labeled with $X_j$, and we have traversed another node labeled with $X_j$ before.
\begin{itemize}
\item If Condition (1) holds, $v_{i+1}$ is set to be an l-node.
      It is clear that the length of the $i$-th segment is exactly $|Y_{a_{i}}|$ and $|Y_{a_{i+1}}| = 2|Y_{a_{i}}|$.
\item If Condition (1) does not hold but Condition (2) holds, $v_{i+1}$ is set to be an r-node.
      Since $v_{i+1}$ contains position $2|Y_{a_{i}}|$, the length of the $i$-th segment is less than $|Y_{a_{i}}|$ and $|Y_{a_{i+1}}| < 2|Y_{a_{i}}|$.
      We remark that since $X_j$ appears in $Y_{a_{i+1}}$, then $|Y_{a_{i+1}}| + |X_j| \leq 2|Y_{a_{i+1}}|$, 
      and therefore, we never move down $v_{i+1}$ for the segments to follow.
\end{itemize}
We iterate the above procedures until we obtain a prefix variable $Y_{a_{k-1}}$ that satisfies $|X_n| \leq 2|Y_{a_{k-1}}|$.
We let $u_{k}$ be the deepest node on the path from the root to $v_{k-1}$ which contains position $|s|$, 
and let $v_{k}$ be the right child of $u_{k}$.
Since $|Y_{a_{i}}| < 2|Y_{a_{i+2}}|$ for any $1 \leq i < k$, $k = O(\log N)$ holds.

We note that the $i$-th segment can be represented by the concatenation of ``inner'' nodes attached to the path from $v_{i}$ to $v_{i+1}$,
and hence, the number of new variables needed for representing the segment is bounded by the number of such nodes.
Consider all the edges we have traversed in the derivation tree $T_n$ of $X_n$.
Each edge contributes to at most one new variable for some segment (see also Fig.~\ref{fig:approx_doubling} in Appendix C).
Since each variable $X_j$ is used constant times for moving down due to Condition (2),
the number of the traversed edges as well as $n'$ is $O(n)$.
Also, it is easy to make the height of $Y_{b_{i}}$ be $O(\height(\mathcal{S}))$ for any $1 \leq i < k$.
Thus $O(\height(\mathcal{S'})) = O(\log N + \height(\mathcal{S})) = O(\height(\mathcal{S}))$.
\qed
\end{proof}

\subsection{Improved algorithm}

Using Lemmas~\ref{lem:LCE} and~\ref{lem:approx_doubling}, we get the following theorem.
\begin{theorem}
Given an SLP $\mathcal{S}$ of size $n$ and height $h$ that describes string $s$ of length $N$,
an $O(n \log N)$-size compact representation of all runs in $s$ can be computed in $O(n^3h)$ time and $O(n^2)$ working space.
\end{theorem}
\begin{proof}
Using Lemma~\ref{lem:LCE}, we first pre-process $\mathcal{S}$ in $O(n^2h)$ time 
so that any ``right-right'' or ``left-left'' $\LCE$ query can be answered in $O(h \log N)$ time.
For each variable $X_i \rightarrow X_\ell X_r$, 
using Lemma~\ref{lem:approx_doubling}, 
we build temporal SLPs $\mathcal{T}$ and $\mathcal{T}'$ which have respectively 
approximately doubling suffix variables of $X_\ell$ and prefix variables of $X_r$, 
and compute two AP-tables for $\mathcal{S}$ and each of them in $O(n^2h)$ time.
For each of the $O(\log N)$ prefix/suffix variables, 
we use it as a seed and find all corresponding runs by using $\LS$ and $\LCE$ queries 
constant times.
Hence the time complexity is $O(n^2h + n(n^2h + (h + h \log N) \log N)) = O(n^3h)$.
The space requirement is $O(n^2)$, the same as the basic algorithm.
\qed
\end{proof}

%
%

\section{Finding $g$-gapped palindromes}

A similar strategy to finding runs on SLPs can be used for 
computing a compact representation of the set $\gpals(s)$ of 
$g$-gapped palindromes from an SLP $\mathcal{S}$ that describes string $s$.
As in the case of runs, 
we add two auxiliary variables $X_{n+1} \rightarrow X_{\$} X_n$ and $X_{n+2} \rightarrow X_{n+1} X_{\$^{\prime}}$.
For each production $X_i \rightarrow X_{\ell} X_{r}$ with $i \leq n+2$,
we consider the set $\kushigpals(X_i)$ of $g$-gapped palindromes which touch or cross the boundary of 
$X_i$ and are completed in $X_i$,
i.e., those that are not prefixes nor suffixes of $X_i$.
Formally, 
\[
  \kushigpals(X_i) = \{\gpal{b}{e} \in \gpals(X_i) \mid 1 \leq b-1 \leq |X_{\ell}| < e+1 \leq |X_i| \}.
\]

Each $g$-gapped palindrome in $X_i$ can be divided into three groups (see also Fig.~\ref{fig:cases_gpal}); 
(1) its right arm crosses or touches with its right end the boundary of $X_i$,
(2) its left arm crosses or touches with its left end the boundary of $X_i$,
(3) the others.

\begin{figure*}[t]
  \begin{center}
    \includegraphics[scale=0.70]{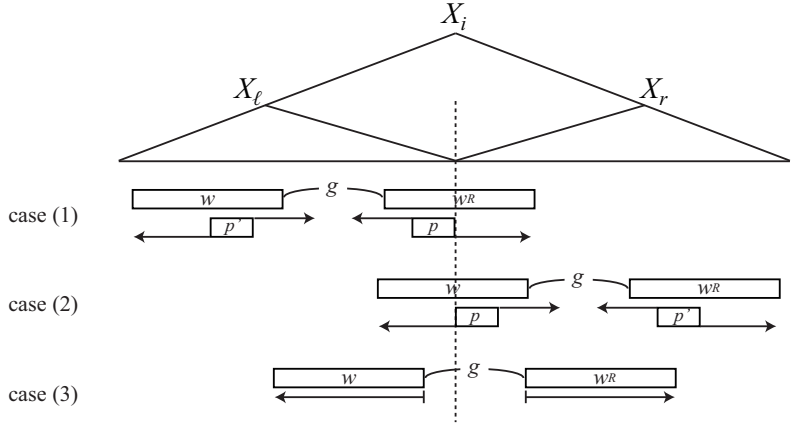}
  \end{center}
  \caption{
    Three groups of $g$-gapped palindromes to be found in $X_i$.
  }
  \label{fig:cases_gpal}
\end{figure*}

For Case (3), for every $|X_{\ell}|-g+1 \leq j < |X_{\ell}|$ we check if $\LCP(\rev{X_i[1..j]}, X_i[j+g+1..|X_i|]) > 0$ or not.
From Lemma~\ref{lem:LCE}, it can be done in $O(g h \log N)$ time for any variable by using ``left-right'' $\LCE$ (excluding pre-processing time for $\LCE$).
Hence we can compute all such $g$-gapped palindromes for all productions in $O(n^2h + g n h \log N)$ time,
and clearly they can be stored in $O(ng)$ space.

For Case (1), let $w_\ell$ be the prefix of the right arm 
which is also a suffix of $\derive(X_\ell)$.
We take approximately doubling suffixes of $X_\ell$ as seeds.
Let $p$ be the longest seed that is contained in $w_\ell$.
We can find $g$-gapped palindromes by the following steps:
\begin{description}
  \item[Step 1:] Conduct local search of $p' = \rev{p}$ in an ``appropriate range'' of $X_{i}$
        and find it in the left arm of palindrome. \label{enum:matching_gpal}
  \item[Step 2:] Compute ``right-left'' $\LCE$ of $p'$ and $p$, then check that the gap can be $g$.
        The outward maximal extension can be obtained by computing ``left-right'' $\LCE$ queries 
on the occurrences of $p'$ and $p$. \label{enum:lcp_gpal}
\end{description}
As in the case of runs, 
for each seed, the length of the range where the local search is performed
in Step 1 is only $O(|p|)$.
Hence, the occurrences of $p'$ can be represented by a constant number of arithmetic progressions.
Also, we can obtain $O(1)$-space representation of $g$-gapped palindromes 
for each arithmetic progression representing overlapping occurrences of $p'$, 
by using a constant number of $\LCE$ queries.
Therefore, by processing $O(\log N)$ seeds for every variable $X_i$,
we can compute in $O(n^2h + n(n^2h + (h + h \log N) \log N)) = O(n^3h)$ time 
an $O(n \log N)$-size representation of all $g$-gapped palindromes for Case (1) in $s$.

In a symmetric way of Case (1), we can find all $g$-gapped palindromes for Case (2).
Putting all together, we get the following theorem.

\begin{theorem}\label{theorem:compute_gap_pals}
Given an SLP of size $n$ and height $h$ that describes string $s$ of length $N$, and non-negative integer $g$,
an $O(n \log N + ng)$-size compact representation of all $g$-gapped palindromes in $s$ can be computed in $O(n^3h + gnh\log N)$ time and $O(n^2)$ working space.
\end{theorem}

\section{Discussions}

Let $\mathbb{R}$ and $\mathbb{G}$ denote the output compact representations of 
the runs and $g$-gapped palindromes of a given SLP $\mathcal{S}$, respectively,
and let $|\mathbb{R}|$ and $|\mathbb{G}|$ denote their size.
Here we show an application of $\mathbb{R}$ and $\mathbb{G}$;
given any interval $[b, e]$ in $s$, 
we can count the number of runs and gapped palindromes in $s[b..e]$ in $O(n + |\mathbb{R}|)$ 
and $O(n + |\mathbb{G}|)$ time, respectively.
We will describe only the case of runs, 
but a similar technique can be applied to gapped palindromes.
As is described in Section~\ref{sec:lce}, $s[b..e]$ can be represented by 
a sequence $U = (X_{u(1)}, X_{u(2)}, \ldots, X_{u(m)})$ of $O(h)$ variables of $\mathcal{S}$.
Let $\mathcal{T}$ be the SLP obtained by concatenating the variables of $U$.
\sinote*{please check}{%
There are three different types of runs in $\mathbb{R}$:
(1) runs that are completely within the subtree rooted at one of the nodes of $U$;
(2) runs that begin and end inside $[b, e]$ and cross or touch any border between consecutive nodes of $U$;
(3) runs that begin and/or end outside $[b, e]$.
Observe that the runs of types (2) and (3)
cross or touch the boundary of one of the nodes in the path from the root to the $b$-th leaf of the derivation tree $T_\mathcal{S}$, 
or in the path from the root to the $e$-th leaf of $T_\mathcal{S}$.
A run that begins outside $[b, e]$ is counted 
only if the suffix of the run that intersects $[b, e]$ has an exponent of at least 2.
The symmetric variant applies to a run that ends outside $[b, e]$.
Thus, the number of runs of types (2) and (3) can be counted in $O(n + 2|\mathbb{R}|)$ time.
Since we can compute in a total of $O(n)$ time
the number of nodes in the derivation tree of $\mathcal{T}$
that are labeled by $X_i$ for all variables $X_i$,
the number of runs of type (1) for all variables $X_{u(j)}$ 
can be counted in $O(n + |\mathbb{R}|)$ time.
Noticing that runs are compact representation of squares, 
we can also count the number of occurrences of all squares 
in $s[b..e]$ in $O(n + |\mathbb{R}|)$ time
by simple arithmetic operations.
}%

The approximate doubling and $\LCE$ algorithms of 
Section~\ref{sec:repetitions} can be used as basis of other efficient algorithms on SLPs.
For example, using approximate doubling,
we can reduce the number of pairs of variables for which the AP-table has to be computed
in the algorithms of Lifshits~\cite{lifshits07:_proces_compr_texts}, 
which compute compact representations of all periods and covers of a string given as an SLP.
As a result, we improve the time complexities 
from $O(n^2h \log N)$ to $O(n^2h)$ for periods, 
and from $O(n^2h \log^2 N)$ to $O(nh(n+\log^2 N))$ for covers.

\bibliographystyle{splncs03}
\bibliography{ref}

\begin{thebibliography}{10}
\providecommand{\url}[1]{\texttt{#1}}
\providecommand{\urlprefix}{URL }

\bibitem{AB96}
Apostolico, A., Breslauer, D.: An optimal {$\Oh{\log \log N}$}-time parallel
  algorithm for detecting all squares in a string. SIAM Journal on Computing
  25(6),  1318--1331 (1996)

\bibitem{ApostolicoBG95}
Apostolico, A., Breslauer, D., Galil, Z.: Parallel detection of all palindromes
  in a string. Theor. Comput. Sci.  141(1{\&}2),  163--173 (1995)

\bibitem{Bannai2012eat}
Bannai, H., Gagie, T., I, T., Inenaga, S., Landau, G.M., Lewenstein, M.: An
  efficient algorithm to test square-freeness of strings compressed by
  straight-line programs. Inf. Process. Lett.  112(19),  711--714 (2012)

\bibitem{philip11:_random_acces_gramm_compr_strin}
Bille, P., Landau, G.M., Raman, R., Sadakane, K., Satti, S.R., Weimann, O.:
  Random access to grammar-compressed strings. In: Proc. SODA 2011. pp.
  373--389 (2011)

\bibitem{CrochemoreIR09}
Crochemore, M., Ilie, L., Rytter, W.: Repetitions in strings: Algorithms and
  combinatorics. Theor. Comput. Sci.  410(50),  5227--5235 (2009)

\bibitem{crochemore91:_effic}
Crochemore, M., Rytter, W.: Efficient parallel algorithms to test
  square-freeness and factorize strings. Information Processing Letters  38(2),
   57 -- 60 (1991)

\bibitem{jansson07:_onlin_dynam_recog_squar_strin}
Jansson, J., Peng, Z.: Online and dynamic recognition of squarefree strings.
  International Journal of Foundations of Computer Science  18(2),  401--414
  (2007)

\bibitem{Khvorost2012CAS}
Khvorost, L.: Computing all squares in compressed texts. In: Proceedings of the
  2nd Russian Finnish Symposium on Discrete Mathemtics. vol.~17, pp. 116--122
  (2012)

\bibitem{KolpakovK99}
Kolpakov, R.M., Kucherov, G.: Finding maximal repetitions in a word in linear
  time. In: FOCS. pp. 596--604 (1999)

\bibitem{lifshits07:_proces_compr_texts}
Lifshits, Y.: Processing compressed texts: A tractability border. In: Proc. CPM
  2007. LNCS, vol. 4580, pp. 228--240 (2007)

\bibitem{Main1989Dlm}
Main, M.G.: Detecting leftmost maximal periodicities. Discrete Applied
  Mathematics  25(1-2),  145--153 (1989)

\bibitem{MainLorentz84}
Main, M.G., Lorentz, R.J.: An {$\Oh{n \log n}$} algorithm for finding all
  repetitions in a string. Journal of Algorithms  5(3),  422--432 (1984)

\bibitem{Manacher75}
Manacher, G.K.: A new linear-time ``on-line'' algorithm for finding the
  smallest initial palindrome of a string. J. ACM  22(3),  346--351 (1975)

\bibitem{matsubara_tcs2009}
Matsubara, W., Inenaga, S., Ishino, A., Shinohara, A., Nakamura, T., Hashimoto,
  K.: Efficient algorithms to compute compressed longest common substrings and
  compressed palindromes. Theoretical Computer Science  410(8--10),  900--913
  (2009)

\bibitem{MST97}
Miyazaki, M., Shinohara, A., Takeda, M.: An improved pattern matching algorithm
  for strings in terms of straight-line programs. In: Proceedings of the 8th
  Annual Symposium on Combinatorial Pattern Matching. pp. 1--11 (1997)

\end{thebibliography}

\newpage
\section*{Appendix A: Details of the algorithm to find runs}

In this section, we describe how we process occurrences of $p'$ at Step 2 of the basic algorithm.
To handle occurrences of $p'$ that are represented by an arithmetic progression, 
we make use of its periodicity.

For any string $s$ and positive integer $c \leq |s|$, 
let $\repr{s}{c}$ (resp. $\repl{s}{c}$) denote the length of 
the longest prefix (resp. suffix) of $s$ having period $c$.

\begin{lemma}\label{lem:handle_ap}
Let $s,p \in \Sigma^+$ and 
$\{a_0, a_1, \dots, a_k\}$ be consecutive occurrences of $p$ in $s$ 
that form a single arithmetic progression with common difference $c \leq |p|$.
Let $z_j = s[a_j + |p|..|s|]$ and $z'_j = s[1..a_j - 1]$ for any $0 \leq j \leq k$.
For any non-empty strings $x, x' \in \Sigma^+$,
it holds that 
\begin{eqnarray*}
\LCP(z_j, x) & = & 
  \begin{cases}
    \min\{\overr{\alpha}-cj, \overr{\beta}\} & \mbox{ if } \overr{\alpha}-cj \neq \overr{\beta},\\
    \overr{\beta} + \LCP(z_0[\overr{\beta}+1..|z_0|], x[\overr{\beta}+1..|x|]) & \mbox{ otherwise, and}
  \end{cases}\\
\LCS(z'_j, x') & = &
  \begin{cases}
    \min\{\overl{\alpha}+cj, \overl{\beta}\} & \mbox{ if } \overl{\alpha}+cj \neq \overl{\beta},\\
    \overl{\beta} + \LCS(z'_0[1..|z'_0|-\overl{\beta}], x'[1..|x'|-\overl{\beta}]) & \mbox{ otherwise,}
  \end{cases}
\end{eqnarray*}
where $\overr{\alpha} = \repr{pz_0}{c} - |p|$, $\overr{\beta} = \repr{px}{c} - |p|$, $\overl{\alpha} = \repl{z'_0p}{c} - |p|$ and $\overl{\beta} = \repl{x'p}{c}$.
\end{lemma}
\begin{proof}
Since $\repr{pz_j}{c} = \overr{\alpha}-cj + |p|$, 
both $pz_j$ and $px$ have a prefix of length $\min\{\overr{\alpha}-cj, \overr{\beta}\} + |p|$ with period $c$ (see also Fig.~\ref{fig:handle_ap}).
If $\overr{\alpha}-cj \neq \overr{\beta}$, either $pz_j$ or $px$ has a prefix of length $\min\{\overr{\alpha}-cj, \overr{\beta}\} + |p| + 1$ with period $c$
while the other does not,
and hence $\LCP(z_j, x) = \LCP(pz_j, px) - |p| = \min\{\overr{\alpha}-cj, \overr{\beta}\}$.
Only when the period breaks the periodicity, 
i.e., $\overr{\alpha}-cj = \overr{\beta}$, $\LCP(z_j, x)$ could expand.
Note that such expansion occurs at most once.
Similarly, since $\repl{z'_jp}{c} = \overl{\alpha}+cj$ we get the statement for $\LCS(z'_j, x')$.
\qed
\end{proof}

\begin{figure*}[hp]
  \begin{center}
    \includegraphics[scale=0.65]{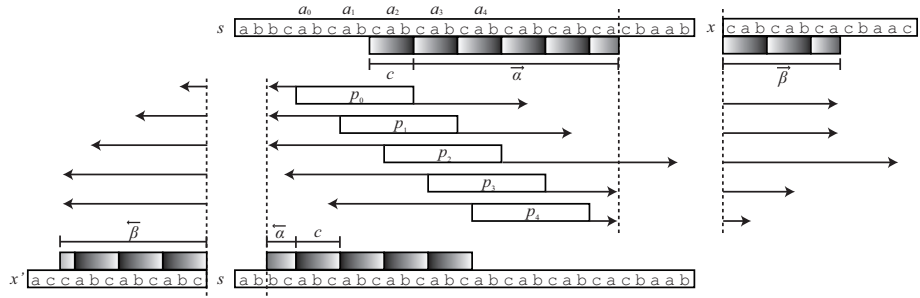}
  \end{center}
  \caption{
    Illustration for Lemma~\ref{lem:handle_ap}.
  }
  \label{fig:handle_ap}
\end{figure*}

In the next lemma, we show how to handle one of the arithmetic progressions 
computed in Step 2 of Case (3).
\begin{lemma} \label{lem:SLP_runs}
Let $X_i \rightarrow X_\ell X_r$ be a production of an SLP of size $n$ 
and $p$ be the suffix of $\derive(X_\ell)$ of length $2^{t-1}$.
Let $\{a_0, a_1, \dots, a_k\}$ be consecutive occurrences of $p'$ in $\derive(X_i)$
which form a single arithmetic progression,
which are computed in Step 2 of Case (3).
We can detect all runs corresponding to the occurrences of $p'$ by using $\LCE$ constant times.
Also, such runs are represented in constant space.
\end{lemma}
\begin{proof}
We apply Lemma~\ref{lem:handle_ap} 
by letting $s = \derive(X_i)$, $x = \derive(X_r)$ and $x' = \derive(X_\ell)[1..|\derive(X_\ell)|-|p|]$.
First we compute $\overr{\alpha} = \LCP(pz_0, p[c+1..|p|]z_0) + c - |p|$, $\overr{\beta} = \LCP(px, p[c+1..|p|]x) + c - |p|$, 
$\overl{\alpha} = \LCS(z'_0p, z'_0p[1..|p| - c]) + c - |p|$ and $\overl{\beta} = \LCS(x'p, x'p[1..|p| - c]) + c - |p|$ by using $\LCP$ and $\LCS$ four times.

\begin{claim}
If $\overr{\beta} + \overl{\alpha} \geq a_0 - 1 + c$, the root of any repetition detected from $a_j$ is not primitive.
\end{claim}
\begin{paragraph} 
{\itshape Proof of Claim.}
If $\overr{\beta} + \overl{\alpha} \geq a_0 - 1 + c$, $pyp$ must have period $c$, where $y$ is the prefix of length $a_1 - 1$ of $x$.
Since $pyp[c+1..c+p] = p$, $|yp|-c$ is a period of $yp$.
It follows from the periodicity lemma that $py$, as well as every $a_j + |p| - 1$, is divisible by greatest common divisor of $c$ and $|yp|-c$,
and hence the root of any repetition detected from $a_j$ is not primitive.
\qed
\end{paragraph}
\vspace*{5pt}

From the above claim, in what follows we assume that 
$\overr{\beta} + \overl{\alpha} < a_0 - 1 + c$.
Let $d_j = a_j - 1 + |p| = a_0 - 1 + |p| + cj$, and then
we want to check if $\LCP(z_j, x) + \LCS(z'_j, x') \geq d_j - |p| = a_0 - 1 + cj$,
or equivalently, $\LCP(z_j, x) + \LCS(z'_j, x') - cj \geq a_0 - 1$.

Let $j' = \min\{j \geq 0 \mid \overr{\alpha}-cj \leq \overr{\beta}\}$ and $j'' = \min\{j \geq 0 \mid \overl{\alpha}+cj \geq \overl{\beta}\}$.
For any $0 \leq j < \min\{j', j''\}$, it follows from $\LCP(z_j, x) = \overr{\beta}$ and $\LCS(z'_j, x') = \overl{\alpha}+cj$ that
$\LCP(z_j, x) + \LCS(z'_j, x') - cj = \overr{\beta} + \overl{\alpha}$, 
and hence a repetition $\triple{\delta_1-cj}{\delta_2+cj}{\delta_3+cj}$ appears iff $\overr{\beta} + \overl{\alpha} \geq a_0 - 1$,
where $\delta_1 = |x'|+1-\overl{\alpha}, \delta_2 = a_0+|p|+\overr{\beta}-1$ and $\delta_3 = a_0+|p|-1$ are constants.

We show that the root of such repetition $\triple{\delta_1-cj}{\delta_2+cj}{\delta_3+cj}$ is primitive.
Assume on the contrary that it is not primitive, namely, $s^\prime = s[\delta_1-cj..\delta_2+cj] = u^q$ with $|u| \leq (\delta_3+cj)/2$ and $q \geq 4$.
Evidently, $\repr{s^\prime}{c} = \overr{\beta} + \LCS(z'_j, x') + |p| = \overr{\beta} + \overl{\alpha} + |p| + cj$.
It follows from $a_0 - 1 \leq \overr{\beta} + \overl{\alpha} < a_0 - 1 + c$ that $\delta_3+cj \leq \repr{s^\prime}{c} < \delta_3+cj+c < |s^\prime|$.
Since $2|u| \leq \repr{s^\prime}{c}$ and $c \leq |p| \leq (\delta_3+cj)/2 \leq \delta_3+cj - |u| \leq \repr{s^\prime}{c} - |u|$, 
$\repr{s^\prime[1..|s^\prime|-|u|}{c} = \repr{s^\prime[|u|+1..|s^\prime|]}{c} + |u|$,
however both $s^\prime[1..|s^\prime|-|u|]$ and $s^\prime[|u|+1..|s^\prime|]$ are $u^{q-1}$, a contradiction.
Therefore, for all $0 \leq j < \min\{j', j''\}$, $\triple{\delta_1-cj}{\delta_2+cj}{\delta_3+cj}$ are runs, 
and they can be encoded by a quintuplet $\quintuplet{\delta_1}{\delta_2}{\delta_3}{c}{\min\{j', j''\}}$.
For any $\min\{j', j''\} \leq j \leq k$ except for $j = j'$ or $j''$, 
$\LCP(z_j, x) + \LCS(z'_j, x') - cj$ is monotonically decreasing by at least $c$ and 
satisfies $\LCP(z_j, x) + \LCS(z'_j, x') - cj < \overr{\beta} + \overl{\alpha} - c < a_0 - 1$,
and hence, no repetition appears.
For $j'$ and $j''$, we can check whether these two occurrences become runs or not by using $\LCE$ constant times.
\qed
\end{proof}

\begin{figure*}[t]
  \begin{center}
    \includegraphics[scale=0.65]{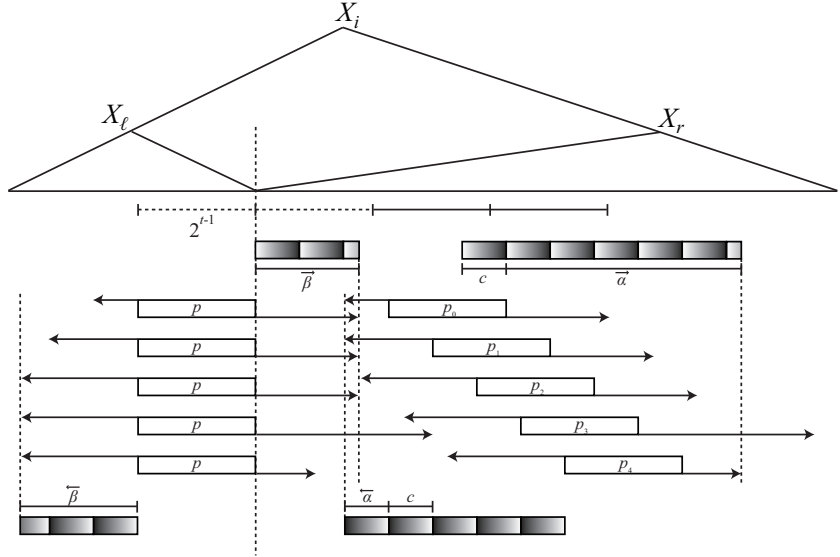}
  \end{center}
  \caption{
    Illustration for Lemma~\ref{lem:SLP_runs}.
    Four runs are found. 
    Here $j' = 3$ and $j'' = 2$.
    The runs from $p_0$ and $p_1$ are encoded by a quintuplet.
    For each $j'$ and $j''$, the run is separately encoded by a quintuplet that shows a single run.
  }
  \label{fig:overlap_run}
\end{figure*}

The other cases can be processed in a similar way.

A minor technicality is that we may redundantly find the same run in different cases.
However, we can avoid duplicates by simply looking into the currently computed runs when we add new runs, spending $O(\log N)$ time.
Also, we can remove repetitions whose root are not primitive by just choosing the smallest period among the repetitions with the same interval.

\newpage

\section*{Appendix B: Proof of Lemma~\ref{lem:fm}}

\begin{proof}
The outline of our algorithm to compute $\FirstMismatch$ follows~\cite{MST97}
which used a slower algorithm for $\Match$.
Assume $|X_i|-k+1 \leq |X_j|$ holds.

If $X_j \rightarrow a$ with $a \in \Sigma$,
then 
\[
 \FirstMismatch(X_i, X_j, k) = 
  \begin{cases}
   1 & \mathrm{if} \ \Match(X_i, X_j, k) = \mathrm{true}, \\ 
   0 & \mathrm{if} \ \Match(X_i, X_j, k) = \mathrm{false}. \\ 
  \end{cases}
\]

If $X_j \rightarrow X_{\ell(j)} X_{r(j)}$,
then we can recursively compute $\FirstMismatch(X_i, X_j, k)$ as follows:
\begin{eqnarray}
\lefteqn{\FirstMismatch(X_i, X_j, k)} \nonumber \\
& = &
   \begin{cases}
    \FirstMismatch(X_i, X_{r(j)}, k + |X_{\ell}|) & \mathrm{if} \ \Match(X_i, X_{\ell(j)}, k) = \mathrm{true}, \\
    \FirstMismatch(X_i, X_{\ell(j)}, k) & \mathrm{if } \ \Match(X_i, X_{\ell(j)}, k) = \mathrm{false}.
   \end{cases} \label{eqn:fm}
\end{eqnarray}

We apply Lemma~\ref{lem:match} to $\mathcal{S}$, 
pre-processing SLP $\mathcal{S}$ in $O(n^2 h)$ time and $O(n^2)$ space,
so that query $\Match(X_i, X_{j^\prime}, k^\prime)$ is answered in $O(\log N)$ time
for any variable $X_{j^\prime}$ and integer $k^\prime$.
Note that in either case of Equation~\ref{eqn:fm},
the height of the second variable decreases by 1.
Hence we can compute $\FirstMismatch(X_i, X_j, k)$ in $O(h \log N)$ time,
after the $O(n^2 h)$-time $O(n^2)$-space pre-processing.
\qed
\end{proof}

\newpage

\section*{Appendix C: Figures}

\begin{figure*}[hp]
  \begin{center}
    \includegraphics[scale=0.6]{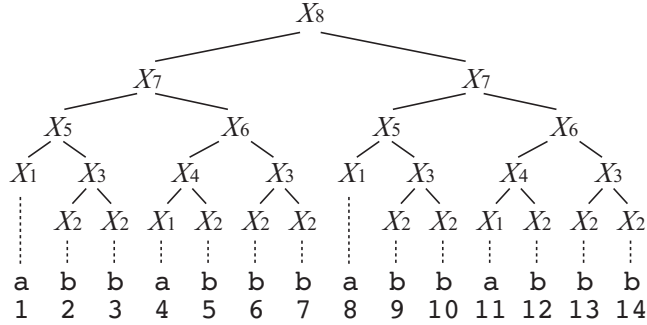}
  \end{center}
  \caption{
    The derivation tree of
    SLP $\mathcal{S} = \{ X_1 \rightarrow \mathtt{a}$, $X_2 \rightarrow \mathtt{b}$, 
    $X_3 \rightarrow X_2 X_2$, $X_4 \rightarrow X_1 X_2$, $X_5 \rightarrow X_1 X_3$,
    $X_6 \rightarrow X_4 X_3$, $X_7 \rightarrow X_5 X_6$, $X_8 \rightarrow X_7 X_7$ \},
    representing string $s = \mathtt{abbabbbabbabbb}$.
  }
  \label{fig:SLP}
\end{figure*}

\begin{figure}[hp]
 \centerline{\includegraphics[width=0.6\textwidth]{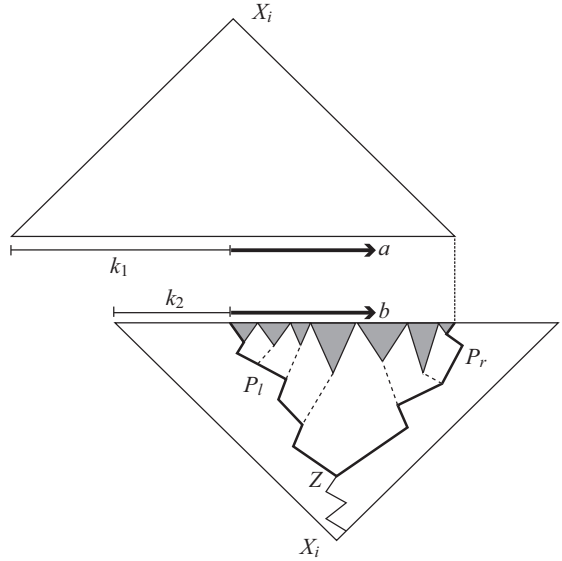}}
 \caption{
  Lemma~\ref{lem:LCE}: Illustration for computing $\LCE(X_i, k_1, k_2)$.
  The roots of the gray subtrees are labeled by the variables in $U$.
  We find the first variable $X_{u(p^\prime)}$ in the list $U$ with which
  the $\Match$ query returns false.
  We then perform the $\FirstMismatch$ query for  
  $X_i$ and $X_{u(p^\prime)}$ using the appropriate offset.
 }
 \label{fig:lce}
\end{figure}

\begin{figure}[hp]
 \centerline{\includegraphics[width=1.0\textwidth]{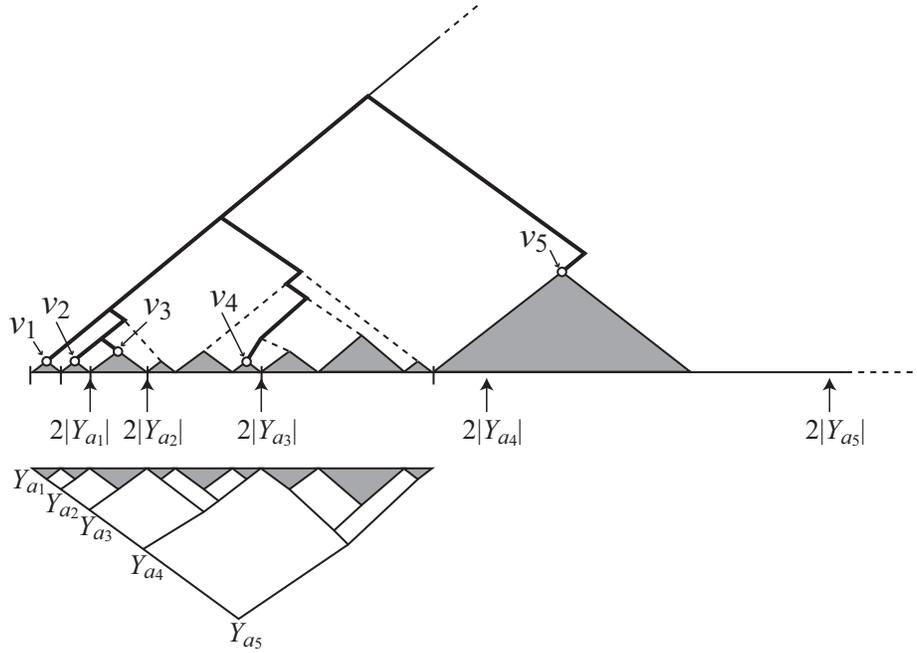}}
 \caption{
  Lemma~\ref{lem:approx_doubling}: Illustration for approximate doubling.
  The prefix variables up to $Y_{a_5}$ have been created.
  The traversals for $v_2$, $v_3$, $v_4$ end due to Condition 1 and that for $v_5$ ends due to Condition 2.
  Each traversed edge (depicted in bold) contributes to at most one new variable for some segment.
  Next, we will resume the traversal from $v_5$ targeting position $2|Y_{a_5}|$, and iterate the procedure until we get the last variable $Y_{a_k}$.
  The total number of bold edges can be bounded by $O(n)$ thanks to Condition 2.
 }
 \label{fig:approx_doubling}
\end{figure}

\end{document}